\documentclass[runningheads]{llncs}

\usepackage[T1]{fontenc}
\def\doi#1{\href{https://doi.org/\detokenize{#1}}{\url{https://doi.org/\detokenize{#1}}}}

\usepackage{amsmath}
\usepackage{amsfonts}
\usepackage{amssymb}
\usepackage{bussproofs}
\usepackage{mathdots}
\usepackage{color}
\usepackage{proof}
\usepackage{bussproofs}
\usepackage{thmtools}
\usepackage{xspace}

\usepackage{enumerate}

\newcommand{\imp}{\supset}
\newcommand{\coimp}{\prec}
\newcommand{\Sa}{\Rightarrow} 
\newcommand{\Ra}{\Sa}

\newcommand{\BiInt}{\mathbf{BiInt}}

\newcommand{\SF}{\mathbf{S5}}
\newcommand{\LK}{\mathbf{LK}}

\newcommand{\cut}{cut}

\newcommand{\weak}{w}

\newcommand{\contr}{contr}
\newcommand{\hyp}{init}

\newcommand{\leftpremise}{\delta_1}
\newcommand{\rightpremise}{\delta_2}

\begin{document}

\title{A theory of cut-restriction: first steps}

\author
{
Agata Ciabattoni\inst{1}\orcidID{0000-0001-6947-8772} \and
Timo Lang\inst{2}\orcidID{0000-0002-8257-968X} \and
Revantha Ramanayake\inst{3}\orcidID{0000-0002-7940-9065}
}

\authorrunning{A. Ciabattoni et al.}

\institute{
TU Vienna, Austria \email{agata@logic.at}
\and
University College London, United Kingdom
\email{timo.lang@ucl.ac.uk}
\and
University of Groningen, Netherlands\\
\email{d.r.s.ramanayake@rug.nl}
}

\maketitle 

\begin{abstract}
Cut-elimination is the bedrock of  proof theory. It is the algorithm that eliminates cuts from a sequent calculus proof that leads to cut-free calculi and applications.
Cut-elimination applies to many logics irrespective of their semantics. Such is its influence that whenever cut-elimination is not provable in a sequent calculus the invariable response has been a move to a richer proof system to regain it. 
In this paper we investigate a radically different approach to the latter: adapting age-old cut-elimination to restrict the shape of the cut-formulas when elimination is not possible.
We tackle the ``first level'' above cut-free: analytic cuts.
Our methodology is applied to the sequent calculi for bi-intuitionistic logic and $S5$ where analytic cuts are already known to be required.
This marks the first steps in a theory of cut-restriction.

\end{abstract}


\section{Introduction}
Cut-elimination is {\em the} fundamental result of proof theory. At its heart, it is an
algorithm for transforming any proof into an analytic proof, i.e. the only formu-
las that occur in the proof are subformulas of the final statement. This is done by eliminating the most common source of non-analyticity: the cut-rule.
Though larger in size, analytic proofs are better behaved and more amenable to meta-theoretic investigation, especially because the space of proofs under consideration is greatly constrained. Gentzen's motivation in the 1930's was a finitistic proof of the consistency of arithmetic (paving the way for ordinal analysis) but the influence of cut-elimination is far beyond that. The cut-free calculi that are the offspring of cut-elimination are central in structural proof theory where they are used to prove properties of the underlying logic (e.g., consistency, decidability, upper bounds, various flavours of interpolation, and disjunction properties), and appear as semantic tableaux and within automated theorem proving. 


Let us pick up the story in the decades following Gentzen's seminal result.
Cut-elimination was originally proved for the sequent calculi for classical and intuitionistic logic IL. 
The programme of developing cut-free calculi via cut-elimination was soon extended to non-classical logics, notably modal logics. The first significant obstacle was encountered in the early 1950's: how to eliminate cuts in the proof calculus for the modal logic $S5$? 
While syntactic decision arguments for $S5$ were obtained (see Ohnishi and Matsumoto~\cite{OhnMat57,OhnMat59}), it seems that 
the first proof of cut-elimination was by Mints~\cite{Min68} who extended the metalanguage of 
Gentzen's sequent calculus to essentially what is today called the hypersequent calculus. Hypersequent calculi were rediscovered independently in the 1980's, by Pottinger~\cite{Pot83} to obtain cut-elimination for modal logics, and Avron~\cite{Avr87}, who applied them to relevant, modal, and intermediate logics. The floodgates had been opened: when cut-elimination cannot be obtained for a logic of interest, the standard response in structural proof theory is to look for (or invent) a new proof formalism by extending the metalanguage of the sequent calculus such that cut-elimination \textit{does} hold. Most proof formalisms have been introduced in the last three decades in this way; nested, labelled, bunched, linear, tree-hypersequent, display sequent, and many more.

A question that has received little attention in all this time is whether it is possible to adapt cut-elimination to simplify those cuts that are not eliminable. This is {\em cut-restriction}, an idea (and terminology) already introduced\footnote{
There, an algorithm to obtain sequent calculi with restricted cuts (parametric on the end formula but not necessarily subformulas) for many families of logics was introduced by essentially composing two algorithms: the algorithms to transform logic specifications into cut-free  hypersequent calculi~\cite{CiaGalTer08,Lah13}, and the reduction of cut-free hypersequent calculi to sequent calculi with restricted cuts. This is evidence that cut-restriction is \textit{in principle} possible on a broad and ambitious scale.} 
in Ciabattoni et al.~\cite{CiaLanRam21} but not pursued.
It is the subject of this paper.


A result in the spirit of cut-restriction is contained in Takano's 1992 paper~\cite{Tak92} on the analytic cut property (cuts can be restricted to subformulas of the rule conclusion) of $S5$. 
Even there, the argument does not appear to have the generic character that is a hallmark of Gentzen's arguments.
Other results restricting cuts rely on the semantics of the logics under consideration, e.g., \cite{Fit78,FisLad79,KowOno17,AvrLahav,Tak01,BezGhi14,Tak20}.
In contrast, Gentzen's cut-elimination method applies rather uniformly across diverse proof systems irrespective of the logic's semantics. A major motivation for choosing to adapt cut-elimination as the path to restricting cuts is the ambition that it too will be applicable as broadly.

In addition to the theoretical interest in adapting this most fundamental notion in proof theory, there are several further motivations.

Cut-restriction is an alternative to moving to proof formalisms more complex than the sequent calculus to obtain cut-elimination.
New proof formalisms are designed so that the cut-elimination proof goes through, but it is always possible that the same could have been achieved using a simpler formalism.
For example, the hypersequent calculus for $S5$ manipulates arbitrarily many component sequents although it turns out that only two are in fact required~\cite{Ind19}. 
Additional structure in the proof formalism is often a hindrance to proving metalogical results, and for limiting the proof search space.
At this stage we do not intend to compare restricted cuts with additional structure in a proof formalism but we do note that decidability and complexity arguments proceed in the presence of finitely many sequent cuts, and interpolation proofs too,  at least in the case of analytic (or almost analytic) cuts, e.g.,~\cite{Tak92,KowOno17}.
In addition, restricted sequent cuts may also be useful from a computational viewpoint~\cite{DAgoMon94}, and the usual blow-up in the size of cut-free proofs can be alleviated if one admits analytic cuts~\cite{dAgo92}.

In this paper we 
tackle the first ``level'' of cut-restriction above cut-free proofs: 
from arbitrary cuts to analytic cuts. 
The restriction to analytic cuts is of particular conceptual interest because it corresponds directly to the subformula property~\cite{KowOno17}, whereas cut-elimination is a strictly stronger property.
We achieve this restriction by replacing Gentzen's permutation reductions with suitable intermediary cuts. The idea
is first discussed informally in Section~\ref{sec-idea}.
Our method is presented
using as main case study the sequent calculus~\cite{Rau74} for bi-intuitionistic logic---a conservative extension of intuitionistic logic with the connective~$\coimp$ dual to implication. 
The main ingredients for the cut-restriction proof to go through are identified in Section~\ref{attempt} and the argument is adapted to the case of the modal logic $S5$.
%
We emphasise that completeness of these sequent calculi with analytic cuts is already known. For bi-intuitionistic logic it was independently proved using semantics in Kowalski and Ono~\cite{KowOno17} and Avron and Lahav~\cite{AvrLahav}. The problem of establishing this result using proof-theoretic methods---accomplished in this paper---was left open in~\cite{PintoU18} where it was shown that infinitely many cut-formulas of a certain (co)implicational form ensures completeness. For $S5$, the result dates back to Takano~\cite{Tak92} whose proof has inspired our work. 
%
%
The point in this paper is that  we uncover how to obtain these results as an adaptation of age-old cut-elimination. 
These are the first steps in the theory of cut-restriction.

\section{Cut-Elimination adapted to Cut-Restriction: the idea}\label{sec-idea}

Gentzen's cut-elimination argument is well known: stepwise applications of reductions replace cut-rules in the proof by smaller cut-rules with respect to a well-founded relation. An appeal to (transfinite) induction ultimately yields a cut-free proof.
These stepwise reductions come in two flavours: \textit{permutation} and
\textit{principal} reduction. The former shifts a cut one step upwards in either the left premise or the right premise.
Following repeated applications, the situation is reached of a cut in which the cut-formula is \textit{principal} (i.e.
created by the rule immediately above it) in both premises. The principal reduction is now used to replace that cut with cuts on proper subformulas.

The principal reductions depend on the shape of the introduction rules and in some cases they can be hard to find. This is what happens with the modal rule in provability logic $GL$, for example: the change in polarity of the diagonal formula from conclusion to premise necessitates a highly intricate and customised principal reduction~\cite{Val83}.
Here we will consider cut-restriction for sequent calculi with standard rules hence the principal reductions are unproblematic. Therefore we shift our attention to permutation reductions. 

When permutation reductions fail, it is usually because a rule has some condition that is violated when it is shifted from above the cut to below it. The rules in Gentzen's LJ calculus for IL have no such conditions so the permutation reductions are unproblematic. In Maehara's calculus for IL---a multiple-conclusion sequent calculus 
in which \textit{the succedent of the ${\imp}_R$ rule permits no context}, see, e.g.~\cite{Takeuti:87}--- some permutation reductions \textit{do} fail. Consider the following cut
(we use the notation $\cut^*$ to indicate weakening followed by $\cut$):
\begin{equation}
\label{original-cut}
\text{
\AxiomC{$\leftpremise$}
\noLine
\UnaryInfC{$\Gamma \Ra A\imp B,\Delta$}
\AxiomC{$\rightpremise$}
\noLine
\UnaryInfC{$\Gamma, A\imp B, C \Ra D$}
\RightLabel{$\imp_R$}
\UnaryInfC{$\Gamma, A\imp B \Ra C\imp D$}
\RightLabel{$\cut^*$}
\BinaryInfC{$\Gamma \Ra C\imp D,\Delta$}
\DisplayProof
}
\end{equation}
If we attempt to shift this cut upwards in the right premise we get
\begin{center}
\AxiomC{$\leftpremise$}
\noLine
\UnaryInfC{$\Gamma \Ra A\imp B,\Delta$}
\AxiomC{$\rightpremise$}
\noLine
\UnaryInfC{$\Gamma, A\imp B, C \Ra D$}
\RightLabel{$\cut^*$}
\BinaryInfC{$\Gamma, C \Ra D,\Delta$}
\DisplayProof
\end{center}
This does not work because from $\Gamma, C \Ra \Delta, D$ we cannot obtain $\Gamma\Ra \Delta, C\imp D$: the context~$\Delta$ in the former blocks the application of the ${\imp}_R$ rule.
The solution is to repeatedly shift the cut upwards in the left premise $\leftpremise$ until the cut-formula is principal there. 
It is not an issue if we encounter a ${\imp}_R$ in the left premise because this rule 
must be the one that creates the cut-formula. 

There is another way of expressing the composite move: trace the predecessors of the $A\imp B$ in the left premise to identify the inference that creates it. Assume for simplicity that there is just a single such \textit{critical inference}:
\begin{center}
\AxiomC{$\Sigma, A \Ra B$}
\RightLabel{$\imp_R$}
\UnaryInfC{$\Sigma \Ra A\imp B$}
\noLine
\UnaryInfC{$\leftpremise$}
\noLine
\UnaryInfC{$\Gamma \Ra A\imp B,\Delta$}
\AxiomC{$\Gamma, A\imp B, C \Ra D$}
\RightLabel{$\imp_R$}
\UnaryInfC{$\Gamma, A\imp B \Ra C\imp D$}
\RightLabel{$\cut^*$}
\BinaryInfC{$\Gamma \Ra C\imp D,\Delta$}
\DisplayProof
\end{center}
Now shift the cut up to the critical inference to get the following picture.
\begin{center}
\AxiomC{$\Sigma, A \Ra B$}
\RightLabel{$\imp_R$}
\UnaryInfC{$\Sigma \Ra A\imp B$}
\AxiomC{$\Gamma, A\imp B, C \Ra D$}
\RightLabel{$\imp_R$}
\UnaryInfC{$\Gamma, A\imp B \Ra C\imp D$}
\RightLabel{$\cut^*$}
\BinaryInfC{$\Sigma, \Gamma \Ra C\imp D$}
\noLine
\UnaryInfC{$+\leftpremise$}
\noLine
\UnaryInfC{$\Gamma, \Gamma \Ra C\imp D,\Delta$}
\RightLabel{$\contr$}
\UnaryInfC{$\Gamma \Ra C\imp D,\Delta$}
\DisplayProof
\end{center}
Here $+\leftpremise$ indicates that the sequents in the `featured branch' in $\leftpremise$---i.e. the branch from $\Sigma\Ra A\imp B$ to $\Gamma\Ra A\imp B,\Delta$---are enriched with an additional multiset (in this case $\Gamma$) in the antecedent. 
Since Maehara's calculus has no rule with a context restriction in the antecedent, $+\leftpremise$ is well-defined despite the additional multiset.
What next? \textit{Trace the predecessors of} $A\imp B$ in the right premise of cut until $A\imp B$ becomes principal in an inference (once again for simplicity we assume a single such critical inference).
\begin{equation}\label{right-premise-trace}
\text{
\AxiomC{$\Sigma, A \Ra B$}
\RightLabel{$\imp_R$}
\UnaryInfC{$\Sigma \Ra A\imp B$}
\AxiomC{$\Sigma' \Ra A,\Pi$}
\AxiomC{$\Sigma', B \Ra \Pi$}
\RightLabel{$\imp_L$}
\BinaryInfC{$\Sigma', A\imp B \Ra \Pi$}
\noLine
\UnaryInfC{$\rightpremise$}
\noLine
\UnaryInfC{$\Gamma,A\imp B,C\Sa D$}
\RightLabel{$\imp_R$}
\UnaryInfC{$\Gamma, A\imp B \Ra C\imp D$}
\RightLabel{$\cut^*$}
\BinaryInfC{$\Sigma, \Gamma \Ra C\imp D$}
\noLine
\UnaryInfC{$+\leftpremise$}
\noLine
\UnaryInfC{$\Gamma, \Gamma \Ra C\imp D,\Delta$}
\RightLabel{$\contr$}
\UnaryInfC{$\Gamma \Ra C\imp D,\Delta$}
\DisplayProof
}
\end{equation}
Shifting the cut up to the critical inference in the right premise we reach
\begin{center}
\AxiomC{$\Sigma, A \Ra B$}
\RightLabel{$\imp_R$}
\UnaryInfC{$\Sigma \Ra A\imp B$}
\AxiomC{$\Sigma' \Ra A,\Pi$}
\AxiomC{$\Sigma', B \Ra \Pi$}
\RightLabel{$\imp_L$}
\BinaryInfC{$\Sigma', A\imp B \Ra \Pi$}
\RightLabel{$\cut^*$}
\BinaryInfC{$\Sigma, \Sigma' \Ra \Pi$}
\noLine
\UnaryInfC{$+\rightpremise$}
\noLine
\UnaryInfC{$\Sigma, \Gamma, C \Ra D$}
\RightLabel{$\imp_R$}
\UnaryInfC{$\Sigma, \Gamma \Ra C\imp D$}
\noLine
\UnaryInfC{$+\leftpremise$}
\noLine
\UnaryInfC{$\Gamma, \Gamma \Ra C\imp D, \Delta$}
\RightLabel{$\contr$}
\UnaryInfC{$\Gamma\Ra C\imp D, \Delta$}
\DisplayProof
\end{center}
Here again, $+\rightpremise$ indicates that the sequents along the featured branch in $\rightpremise$ are enriched with an additional multiset in the antecedent. 
The cut-formula is principal in both the left and right premise so we use the principal reduction:
\begin{equation}\label{cut-elim-principal}
\text{
%
\AxiomC{$\Sigma' \Ra A,\Pi$}
\AxiomC{$\Sigma, A \Ra B$}
\RightLabel{$\cut^*$}
\BinaryInfC{$\Sigma, \Sigma', \Ra B,\Pi$}
\AxiomC{$\Sigma', B \Ra \Pi$}
\RightLabel{$\cut^*$}
\BinaryInfC{$\Sigma, \Sigma'\Ra \Pi$}
\noLine
\UnaryInfC{$+\rightpremise$}
\noLine
\UnaryInfC{$\Sigma, \Gamma, C \Ra D$}
\RightLabel{$\imp_R$}
\UnaryInfC{$\Sigma, \Gamma \Ra C\imp D$}
\noLine
\UnaryInfC{$+\leftpremise$}
\noLine
\UnaryInfC{$\Gamma, \Gamma \Ra C\imp D,\Delta$}
\RightLabel{$\contr$}
\UnaryInfC{$\Gamma\Ra C\imp D,\Delta$}
\DisplayProof
}
\end{equation}
The transformation from (\ref{original-cut}) to (\ref{cut-elim-principal}) expresses cut-elimination in terms of predecessors of the cut-formulas and principal reduction.
This perspective will help us to adapt the argument to the situation where the parametric reductions are no longer sound and the calculus does not admit cut-elimination. 

We encounter precisely this situation in the well-known sequent calculus~\cite{Rau74} for bi-intuitionistic logic.
This calculus is obtained from the Maehara calculus by extending the language with \textit{coimplication}~$\coimp$ and adding the rules $\coimp_L$ and $\coimp_R$ in Fig.~\ref{fig:BiInt}.
Notice: the \textit{antecedent of the $\coimp_L$ rule permits no context}. From (\ref{cut-elim-principal}) alone we can see that the argument above will no longer work: $+\leftpremise$ is not well-defined as the additional multiset in the antecedent may block an application of $\coimp_L$. 

Our ultimate aim is to
adapt the picure in (\ref{cut-elim-principal}) to obtain a proof of $\Gamma\Ra C\imp D,\Delta$ replacing the non-well-defined derivations by cuts on subformulas of this sequent.
The following adaptation already brings us closer. Here $\land \Sigma$ is the conjunction of all formulas in~$\Sigma$ and $\Sigma \Ra \land \Sigma$ has an easy cut-free proof.
\begin{equation}\label{cut-reduction-simple}
\text{
\AxiomC{$\Sigma \Ra \land \Sigma$}
\noLine
\UnaryInfC{$\leftpremise[\land \Sigma]$}
\noLine
\UnaryInfC{$\Gamma \Ra \Delta, \land \Sigma$}
\AxiomC{$\Sigma' \Ra A,\Pi$}
\AxiomC{$\Sigma, A \Ra B$}
\RightLabel{$\cut^*$}
\BinaryInfC{$\Sigma, \Sigma', \Ra B,\Pi$}
\AxiomC{$\Sigma', B \Ra \Pi$}
\RightLabel{$\cut^*$}
\BinaryInfC{$\Sigma, \Sigma'\Ra \Pi$}
\RightLabel{some $\land_L$'s}
\UnaryInfC{$\land \Sigma, \Sigma'\Ra \Pi$}
\noLine
\UnaryInfC{$+\rightpremise$}
\noLine
\UnaryInfC{$\land \Sigma, \Gamma, C \Ra D$}
\RightLabel{$\imp_R$}
\UnaryInfC{$\land \Sigma, \Gamma \Ra C\imp D$}
\RightLabel{$\cut^*$}
\BinaryInfC{$\Gamma\Ra \Delta, C\imp D$}
\DisplayProof
}
\end{equation}
We denote by $\leftpremise[\land \Sigma]$ the derivation obtained by
replacing in the derivation $\leftpremise$ all the predecessors of the cut-formula $A \imp B$ (including the cut-formula itself) with
$\land \Sigma$. This derivation is sound as the potentially problematic calculus rules are $\imp_R$ and $\coimp_L$ and: a $\imp_R$ application could not have occurred along the featured branch in $\delta_1$ as the traced cut-formula would have blocked it, and placing $\land\Sigma$ in the succedent will not disturb any applications of $\coimp_L$.


The key issue is: what can we say about the cut on $\land \Sigma$? It is not immediate that every formula in $\Sigma$ is a subformula of the endsequent as non-subformulas might trace from cuts below the originally chosen cut, or be an $A\imp B$ that traces from an \textit{analytic cut} (this is a cut where the cut-formula is a subformula of its conclusion c.f. Def.~\ref{def-anal}) on~$A\imp B$ \textit{within} $\delta_{1}$. The latter is a subtle point that arises because we are working in the setting where analytic cuts are not eliminable.
And even if they are subformulas, we cannot expect $\land \Sigma$ to be a subformula of the endsequent. The solution to these issues will---in the general case---require selecting an uppermost non-analytic cut and the introduction of many cuts on the formulas in $\Sigma$ (in fact, exponentially many in $|\Sigma|$), rather than on $\land \Sigma$. 

In summary, our strategy to adapt cut-elimination to cut-restriction is to
retain the principal reductions, replace permutation reductions by tracing the predecessors of the cut-formulas, using only those enriched branches that are well-defined, and applying cut to remove formulas that cannot propagate down.


\section{The case of bi-intuitionistic logic}

\begin{figure}[t]
\fbox{
\begin{minipage}{\textwidth}

Initial rule and cut:
\[
\infer[\hyp\quad\text{($p$ atomic)}]{p\Sa p}
	{}
\hspace{2cm}
\infer[\cut]{\Gamma\Sa\Delta}
	{
	\Gamma\Sa A,\Delta
	&
	\Gamma,A\Sa\Delta
	}
\]
Structural rules:
\[
\infer[\weak]{\Gamma,\Sigma\Sa\Delta,\Pi}
	{
	\Gamma\Sa\Delta
	}
\hspace{2cm}
\infer[\contr]{\Gamma,\Sigma\Sa\Delta,\Pi}
	{
	\Gamma,\Sigma,\Sigma\Sa\Delta,\Delta,\Pi
	}
\]
Logical rules:
\[
\infer[\top_R]{\Sa\top}
	{}
\hspace{2cm}
\infer[\bot_L]{\bot\Sa}
	{}
\]
\[
\infer[\land_L]{\Gamma,A\land B\Sa\Delta}
	{
	\Gamma,A,B\Sa\Delta
	}
\hspace{2cm}
\infer[\land_R]{\Gamma\Sa A\land B,\Delta}
	{
	\Gamma\Sa A,\Delta
	&
	\Gamma\Sa B,\Delta
	}
\]
\[
\infer[\lor_L]{\Gamma,A\lor B\Sa\Delta}
	{
	\Gamma,A\Sa\Delta
	&
	\Gamma,B\Sa\Delta
	}
\hspace{2cm}
\infer[\lor_R]{\Gamma\Sa A\lor B,\Delta}
	{
	\Gamma\Sa A,B,\Delta
	}
\]
\[
\infer[\imp_L]{\Gamma,A\imp B\Sa\Delta}
	{
	\Gamma\Sa A,\Delta
	&
	\Gamma, B\Sa\Delta
	}
\hspace{2cm}
\infer[\imp_R]{\Gamma\Sa A\imp B}
	{
	\Gamma,A\Sa B
	}
\]
\[
\infer[\coimp_L]{A\coimp B\Sa\Delta}
	{
	A\Sa B,\Delta
	}
\hspace{2cm}
\infer[\coimp_R]{\Gamma\Sa A\coimp B,\Delta}
	{
	\Gamma\Sa A,\Delta
	&
	\Gamma,B\Sa \Delta
	}
\]
\end{minipage}}

\caption{The sequent calculus $\BiInt$.}
\label{fig:BiInt}
\end{figure}
We present our method using the sequent calculus $\BiInt$  for bi-intuitionistic logic as a case study (Fig.~\ref{fig:BiInt}).
%
Bi-intuitionistic logic---also known as Heyting–Brouwer logic or subtractive logic---is the conservative extension of intuitionistic logic with a $\coimp$ connective that is dual to implication ({\em coimplication}).
Proposed in~\cite{Rau74} the 
calculus $\BiInt$ extends Maehara's multiple-conclusion calculus for intuitionistic logic (built from sequents of the form $\Gamma\Ra\Delta$ where $\Gamma,\Delta$ are multisets)
with left and right rules for $\coimp$. A counterexample to the admissibility of cut stated in~\cite{Rau74}
(and in~\cite{Crolard04}) was identified by Uustalu and Pinto~\cite{PintoU09}. Failure of cut-elimination prompted a search for analytic calculi for bi-intuitionistic logic
using formalisms extending the meta-language of the sequent calculus~\cite{BuiGor07}, including nested~\cite{GorPosTiu08},  and  labelled sequents~\cite{PintoU09}. Relationships between the calculi were studied in \cite{PintoU18} where the problem of finding a syntactic proof of the completeness of 
$\BiInt$ with analytic cuts  was left open (see~\cite{AvrLahav,KowOno17} for semantic proofs).


 The notion of a \emph{subformula} is defined as usual. $A$ is a \emph{proper subformula} of $B$ if it is a subformula and $A\neq B$. For a set or sequent $X$, we say that $A$ is a subformula of $X$ if $A$ is a subformula of one of the formulas appearing in $X$. 

The rules of the $\BiInt$ calculus in 
Fig.~\ref{fig:BiInt} are actually rule schemata. With an abuse of notation we use $\Gamma, \Delta, \dots$ both for formula multisets and for formula multiset variables. The formulas instantiating $\Gamma$ and $\Delta$ in the rules are called  \emph{context}. We have chosen an \emph{additive} formulation of all rules including cut. This means that contexts are always copied to both premises instead of being distributed among the premises. This is a design choice: in the presence of contraction and weakening, we can freely move between the additive and multiplicative formulations.
Already in Section~\ref{sec-idea} we have adopted the convention of writing $\cut^*$ to abbreviate inferences of $\weak$ followed by $\cut$, e.g.
\begin{center}
    \AxiomC{$\Gamma,\Delta\Sa A,\Sigma,\Pi$}
        \AxiomC{$\Gamma,\Delta',A\Sa \Sigma,\Pi'$}
        \RightLabel{$\cut^*$}
        \BinaryInfC{$\Gamma,\Delta,\Delta'\Sa\Sigma,\Pi,\Pi'$}
    \DisplayProof
\end{center}

The exposed formula in the conclusion of a logical rule is called the \emph{principal formula} of that rule. An \emph{initial sequent} of $\BiInt$ is the conclusion of a rule without premise, i.e. $p\Sa p$, $\Sa\top$ or $\bot\Sa$. These are standard definitions. The \emph{logical constants} $\top$ and $\bot$ are not counted as atomic propositions, so in particular $\top\Sa\top$ and $\bot\Sa\bot$ do not constitute initial sequents.

\begin{remark}
Similar to classical logic, the connectives and rules of $\BiInt$ come in dual pairs: $(\top,\bot),(\land,\lor)$, and $(\imp,\coimp)$. 
Using a formula translation $(\cdot)^\sharp$ that replaces each connective by its dual and inverts the left/right hand side in implications, i.e.\ $(A\imp B)^\sharp=(B^\sharp\coimp A^\sharp)$, $(A\coimp B)^\sharp=(B^\sharp\imp A^\sharp)$  yields that $\Gamma\Sa\Delta$ is provable in $\BiInt$ iff $\Delta^\sharp\Sa\Gamma^\sharp$ is. Due to this symmetry, when proving properties of $\BiInt$ it often suffices to consider just one half of the set of connectives, as the remaining cases are in a precise sense ``the same''---namely, up to $(\cdot)^\sharp$.
\end{remark}

\begin{definition}
A \emph{$\BiInt$-proof} of a sequent $\Gamma\Sa\Delta$ is a tree of sequents rooted in $\Gamma\Sa\Delta$ (the \emph{endsequent}) that is composed of rule instances of $\BiInt$, and such that all leaves are initial sequents of $\BiInt$.
\end{definition}
%
Note that in $\hyp$ we require $p$ to be an atomic formula. By an easy induction on the formula $A$ we can prove completeness of
atomic axioms (\textit{axiom expansion}).

\begin{lemma}\label{lem:genaxioms}
The sequent $A\Sa A$ is derivable in $\BiInt$ without cuts.
\end{lemma}

The following definition of ``analytic cut'' goes back to~\cite{Smu68}.
\begin{definition} 
\label{def-anal}
An instance of the cut rule
\[
\infer[\cut]{\Gamma \Sa\Delta}
	{
	\Gamma\Sa A,\Delta
	&
	\Gamma,A\Sa\Delta
	}
\]
is \emph{analytic} if $A$ is a subformula of $\Gamma\cup\Delta$. A proof is \emph{locally analytic} if every cut in the proof is analytic.
A sequent calculus has the \emph{analytic cut property} if every provable sequent has a locally analytic proof.
\end{definition}


Any calculus whose other rules are analytic\footnote{In the sense that every formula in a premise is a subformula of the endsequent.} has the following property: every formula occurring in a locally analytic proof 
is a subformula of the endsequent. In other words,
the analytic cut property implies the \emph{subformula property}.
This holds, in particular, for $\BiInt$. The nontrivial converse direction was shown in~\cite[Lem. 4.6]{KowOno17} using an elegant induction argument.


\subsection{Proof of Cut-Restriction}
\label{Sec:main}

The analytic cut property for $\BiInt$ is an immediate corollary of the following main theorem which formalises the cut-restriction proof 
described in Section~\ref{sec-idea}.

\begin{theorem}
\label{thm:main}
Every proof in $\BiInt$ that is locally analytic aside from a single non-analytic $cut$ as its lowest inference can be transformed into a locally analytic proof of the same endsequent.
\end{theorem}
\begin{proof}
Let $\delta$ be a proof ending in a non-analytic cut 
\[
\infer[\cut]{\Gamma\Sa \Delta}
	{
	\infer*[\delta_1]{\Gamma\Sa C,\Delta}{}
	&
	\infer*[\delta_2]{\Gamma, C\Sa \Delta}{}
	}
\]
and assume that the subproofs $\delta_1,\delta_2$ are locally analytic. The statement is proved by induction on the size of the cut-formula.

We distinguish cases according to the shape of $C$.

$\blacktriangleright$ $C$ is an atomic formula $p$.
Since $\BiInt$ is consistent, $\Gamma\cup\Delta \not=\emptyset$. Replace \textit{every} occurrence of $p$ 
in $\delta$ by some atomic formula or constant occuring as a subformula in $\Gamma\cup\Delta$. In doing so the lowermost cut becomes analytic, and all cuts in $\delta_1,\delta_2$ remain analytic. The endsequent is unchanged because by assumption the cut on $p$ was non-analytic. If we replaced $p$ by a constant $x\in\{\top,\bot\}$ and in doing so created a non-initial sequent $x\Sa x$ from $p\Sa p$, use $\top_R/\bot_L$ and weakening to obtain an initial sequent.

\emph{Predecessors.} If the cut formula $C$ is not atomic, we will \emph{trace} occurrences of $C$ upwards in the proofs $\delta_1$ and $\delta_2$. Formally, define the  \textit{predecessor relation}\footnote{This is also called the \textit{parametric ancestor relation}.} matching context occurences of $C$ in the conclusion of an inference to the corresponding occurrence(s) in the premise(s). Stop when the occurrence $C$ is principal in the conclusion of logical inference (a \emph{critical inference of $\delta$}), or if $C$ is removed by a weakening rule. We will not encounter an initial sequent $C\Ra C$ because we have assumed that $C$ is not atomic. Note that we consider the occurrence of the cut-formula in the premise of the cut as the first predecessor.

There may be multiple predecessors in a sequent since a predecessor in the conclusion of a contraction rule can become two predecessors in its premise. Clearly all predecessors in $\delta_1$ occur in the succedent while all predecessors in $\delta_2$ occur in the antecedent. 

\newcommand{\irredundance}{\textbf{irredundance}\xspace}
Without loss of generality we will assume the following property:
\begin{description}
    \item[\irredundance:] No sequent in $\delta$ that contains a predecessor of $C$ appears as the conclusion of a cut with cut-formula $C$.
\end{description}
This is justified because any such cut can be replaced by a contraction, e.g.,
\[
\infer[\cut]{\Sigma,C\Sa \Pi}
    {
    \Sigma,C\Sa C,\Pi
    &
    \Sigma,C,C\Sa\Pi
    }
\hspace{0.75cm}\text{to}\hspace{0.75cm}
\infer[\contr]{\Sigma,C\Sa\Pi}
    {
    \Sigma,C,C\Sa\Pi
    }
\]

Moreover for $\delta_*\in\{\delta_1,\delta_2\}$ we denote by $\delta_*[\Sigma]$ the tree (which is not necessarily a proof) resulting from substituting all predecessors of $C$ in $\delta_*$ by the multiset of formulas $\Sigma$. The following properties hold:

\newcommand{\sound}{\textbf{pre-soundness}\xspace}
\newcommand{\reduce}{\textbf{reductivity}\xspace}

\begin{description}
    \item[\sound:] Every inference in $\delta_*$ \emph{with the possible exception of the critical ones} remains a sound inference of the same rule in $\delta_*[\Sigma]$.
    \item[\reduce:] If an analytic cut in $\delta_*$ becomes non-analytic in $\delta_*[\Sigma]$, then that cut is on a \emph{proper} subformula of $C$.
\end{description}

In every sequent that is not the conclusion of a critical inference, the predecessor instantiates the context of the inference. By inspection of the rules, replacing a formula occurring in the context of an inference rule with any multiset preserves the soundness of that inference hence \sound.

If an analytic cut becomes non-analytic then the cut-formula must be a subformula of~$C$. It cannot be $C$ itself because of the irredundance property so it must be a proper subformula and hence \reduce holds.

With these preparations, we now resume the case distinction.

$\blacktriangleright$ $C$ is $\top$. Observe that there are no critical inferences in $\delta_2$ as we do not have a rule for $\top$ in the antecedent, hence $\top$ has been introduced by weakening. By \sound $\delta_2[\emptyset]$ is a $\BiInt$-proof of $\Gamma\Sa\Delta$. Since $\top$ has no proper subformulas, \reduce implies that $\delta_2[\emptyset]$ is locally analytic.

$\blacktriangleright$ $C$ is $\bot$.
Similar to the previous case (by symmetry).

$\blacktriangleright$ $C$ is $A\land B$.
Here we can use a standard invertibility-style argument. Consider the trees $\delta_1[A],\delta_1[B]$ and $\delta_2[A,B]$ with their respective roots $\Gamma\Sa A,\Delta$ and $\Gamma\Sa B,\Delta$ and $\Gamma,A,B\Sa \Delta$. With some minor modifications we can make all the trees into $\BiInt$ proofs (by \sound it suffices to check critical inferences). For example, let 
\[
\infer[\land_R]{\Sigma\Sa A\land B,(A\land B)^*,\Pi}
    {
    \Sigma\Sa A,(A\land B)^*,\Pi
    &
    \Sigma\Sa B,(A\land B)^*,\Pi
    }
\]
be a critical inference  in $\delta_1$. Here $(A\wedge B)^*$ indicates some number of predecessors (possibly none) that are not the principal formula, and $\Sigma,\Pi$ contains no predecessors. Then in $\delta_1[A]$ we have the unsound inference
\[
\infer[unsound]{\Sigma\Sa A,A^*,\Pi}
    {
    \Sigma\Sa A,A^*,\Pi
    &
    \Sigma\Sa B,A^*,\Pi
    }
\]
Nevertheless the above can simply be replaced using the left premise $\Sigma\Sa A,A^*,\Pi$. Similar reasoning applies to $\delta_1[B]$ and $\delta_2[A, B]$. By an abuse of notation, let $\delta_1[A],\delta_1[B]$ and $\delta_2[A,B]$ indicate below the $\BiInt$-proofs where these modifications have been carried out. 
We then obtain the following cut-restriction:
\[
\infer[\cut]{\Gamma\Sa\Delta}
    {
    \infer*[\delta_1{[A]}]{\Gamma\Sa A,\Delta}{}
    &
    \infer[\cut^*]{\Gamma,A\Sa\Delta}
        {
        \infer*[\delta_1{[B]}]{\Gamma\Sa B,\Delta}{}
        &
        \infer*[\delta_2{[A,B]}]{\Gamma,A,B\Sa\Delta}{}
        }
    }
\]
By \reduce, all non-analytic cuts in this proof are on proper subformulas of $A\land B$, and hence eliminable by induction hypothesis.

$\blacktriangleright$ $C$ is $A\lor B$.\\
Similar to the previous case (by symmetry).

$\blacktriangleright$ $C$ is $A\imp B$.


This case is the crucial part of the proof. Define the \textit{critical antecedent} and \textit{critical succedent}
of a critical inference as the set of formulas in its antecedent and succedent
of the conclusion, respectively,
omitting predecessors (if any).
E.g. suppose that the following is a critical inference in $\delta_2$ containing exactly two predecessors in its conclusion.

 \begin{center}
 \AxiomC{$A\imp B, \Sigma \Ra  \Pi, A$}
 \AxiomC{$A\imp B, B, \Sigma \Ra  \Pi, A$}
 \RightLabel{$\imp_L$}
 \BinaryInfC{$A\imp B, A\imp B, \Sigma \Ra \Pi$}
 \DisplayProof
 \end{center}

Then its critical antecedent is $\Sigma$ and its critical succedent is $\Pi$. We remark that due to the shape of $\imp_R$, each critical inference in $\delta_1$ contains exactly one predecessor---the principal formula of $\imp_R$---and the critical succedent is empty:
\begin{center}
\AxiomC{$A, \Sigma \Ra  B$}
\RightLabel{$\imp_R$}
\UnaryInfC{$\Sigma \Ra  A\imp B$}
\DisplayProof
\end{center}
Thus no two critical inferences in $\delta_1$ appear on the same branch.



\newcommand{\critical}{\textbf{tameness}\xspace}
We show the following important property:
\begin{description}
    \item[\critical:] Every formula in the critical antecedent $\Sigma$ of a critical inference in $\delta_1$ is a subformula of $\Gamma\cup \Delta$ or a proper subformula of $A\imp B$.
\end{description}

Since $\delta_1$ is locally analytic by assumption, we know that every formula in it is a subformula of $\Gamma\cup\Delta\cup\{A\imp B\}$. So the only case we have to exclude for \critical is that the critical antecedent contains $A\imp B$. Indeed: Assume towards a contradiction 
that $D= A\imp B \in\Sigma$, and follow the occurences of $D$
downwards in $\delta_1$. Being in the antecedent, $D$ cannot become the $A\imp B$ in the succedent of the endsequent; but also $D$ cannot appear as a subformula of $\Gamma\cup\Delta$ as this would violate the non-analyticity of the lowermost cut. It follows that $D$ must be removed by a cut in $\delta_1$. By \irredundance, this cut cannot be on $D=A\imp B$ itself, so it must be on a formula $D'_{A\imp B}$ containing $A\imp B$ as a proper subformula. Again by local analyticity of $\delta_1$, $D'_{A\imp B}$ must be a subformula of $\Gamma\cup\Delta\cup\{A\imp B\}$. 
Then $D'_{A\imp B}$ must be a subformula of $\Gamma\cup\Delta$ and hence so too is $A\imp B$, contradicting the non-analyticity of the cut in $\delta$.



Now, let $\Sigma_1,\ldots,\Sigma_n$ be all the critical antecedents in $\delta_1$ (i.e. the antecedents of all the critical inferences in~$\delta_1$).

\textbf{Step 1.} For $r\le n$ fix the $r$-th critical inference in $\delta_1$ with subderivation $\delta_1^r$
\begin{center}
\AxiomC{$\vdots\,\delta_1^r$}
\noLine
\UnaryInfC{$\Sigma_r,A \Ra  B$}
\RightLabel{$\imp_R$}
\UnaryInfC{$\Sigma_r \Ra  A\imp B$}
\DisplayProof
\end{center}
and turn to $\delta_2$. Every critical inference in $\delta_2$ has the following form where $A\imp B^{*}$ indicates other predecessors (possibly none) and $\Sigma$ contains no predecessors.
\begin{center}
\AxiomC{$\Sigma, A\imp B^{*} \Ra  A, \Pi$}
\AxiomC{$\Sigma, A\imp B^{*}, B \Ra  \Pi$}
\RightLabel{$\imp_L$}
\BinaryInfC{$\Sigma, A\imp B^{*}, A\imp B \Ra  \Pi$}
\DisplayProof
\end{center}

Obtain the tree $\delta_2[\Sigma_r]$ with root $\Gamma,\Sigma_r\Sa\Delta$ by replacing all critical formulas in $\delta_2$ with the critical antecedent $\Sigma_r$. We need to modify $\delta_2[\Sigma_r]$ into a $\BiInt$-proof. By $\sound$ it suffices to inspect critical inferences in $\delta_2$. In $\delta_2[\Sigma_r]$ each such inference becomes the unsound
\begin{center}
\AxiomC{$\Sigma, \Sigma_r^{*} \Ra  A, \Pi$}
\AxiomC{$\Sigma, \Sigma_r^{*}, B \Ra  \Pi$}
\RightLabel{$unsound$}
\BinaryInfC{$\Sigma, \Sigma_r^*, \Sigma_r \Ra  \Pi$}
\DisplayProof
\end{center}
Using the subderivation $\delta_1^r$ of $\delta_1$ and cuts on $A$ and $B$ we replace this unsound inference in $\delta_2[\Sigma_r]$ by
\begin{equation}\label{principal}
\AxiomC{$\Sigma,  \Sigma_r^* \Ra  A, \Pi$}
\AxiomC{$\vdots\,\delta_1^r$}
\noLine
\UnaryInfC{$A, \Sigma_r\Ra B$}
\RightLabel{$\cut^*$}
\BinaryInfC{$\Sigma, \Sigma_r^*,\Sigma_r, 
\Ra  B, \Pi$}
\AxiomC{$\Sigma, \Sigma_r^{*}, B \Ra  \Pi$}
\RightLabel{$\cut^*$}
\BinaryInfC{$\Sigma, \Sigma_r^*,\Sigma_r\Ra  \Pi$}
\DisplayProof
\end{equation}
In this way we obtain a $\BiInt$-proof of $\Gamma,\Sigma_r\Sa\Delta$. By construction and \reduce all non-analytic cuts in this proof are on proper subformulas of $A\imp B$. So by the induction hypothesis, we obtain a locally analytic proof of
\begin{equation}\label{from-epsilon}
\Gamma,\Sigma_r\Ra  \Delta
\end{equation}

In particular, if $\Sigma_r=\emptyset$ we are already done. For the remainder of the proof let us therefore assume that all critical antecedents $\Sigma_1,\ldots,\Sigma_n$ are nonempty.


\textbf{Step 2.} Let $(D_1,\ldots,D_n)\in \Sigma_1\times\ldots\times\Sigma_n$ and consider the tree $\delta_1[D_1,\ldots,D_n]$ with root $\Gamma\Sa D_1,\ldots,D_n,\Delta$. Again using \sound we make $\delta_1[D_1,\ldots,D_n]$ into a $\BiInt$ proof by replacing any critical inference turned unsound 
\[
\infer[unsound]{\Sigma_r\Sa D_1,\ldots,D_n}
    {
    \Sigma_r,A\Sa B
    }
\]
with a cut-free proof of $\Sigma_r\Ra D_1,\ldots,D_n$, which exists (by Lemma~\ref{lem:genaxioms}) since $D_r\in\Sigma_r$. Again all non-analytic cuts in the resulting $\BiInt$-proof are on proper subformulas of $A\imp B$ by $\reduce$, and so we can use the induction hypothesis to obtain a locally analytic proof of

\begin{equation}\label{from-delta}
\Gamma\Sa D_1,\ldots,D_n,\Delta
\end{equation}

\textbf{Step 3.} We now combine all instances of (\ref{from-epsilon}) and (\ref{from-delta}) using cuts on formulas in $\bigcup_{i=1}^n\Sigma_i$ to obtain a locally analytic proof of $\Gamma\Sa\Delta$.

Enumerate the elements of $\Sigma_r$ as $D_r^1, D_r^2, \ldots, D_r^{m_r}$.
For every $(D_2,\ldots,D_n)\in \Sigma_2\times\ldots\times \Sigma_n$, obtain a derivation of
\begin{equation}\label{next}
\Gamma\Sa D_2,\ldots,D_n,\Delta
\end{equation}
as follows by cuts on the elements of $\Sigma_1$. In the following, the cut-formula has been boxed to make it easier to identify.
\begin{center}
\begin{footnotesize}
\AxiomC{\hspace{2cm}(\ref{from-delta})\hspace{-2cm}}
\noLine
\UnaryInfC{$\Gamma\Ra \text{\fbox{$D_1^{m_{1}}$}} D_2 \ldots D_n, \Delta$ \hspace{-4cm}}
\AxiomC{(\ref{from-delta})}
\noLine
\UnaryInfC{$\Gamma\Ra \text{\fbox{$D_1^2$}} D_2 \ldots D_n, \Delta$}
\AxiomC{(\ref{from-delta})}
\noLine
\UnaryInfC{$\Gamma\Ra \text{\fbox{$D_1^1$}} D_2 .. D_n, \Delta$}
\AxiomC{(\ref{from-epsilon})}
\noLine
\UnaryInfC{$\text{\fbox{$D_1^1$}} D_1^2 ..
D_1^{m_{1}}, \Gamma\Ra \Delta$}
\RightLabel{$\cut^*$}
\BinaryInfC{$\text{\fbox{$D_1^2$}}, D_1^3 \ldots D_1^{m_{1}}, \Gamma\Ra D_2 \ldots D_n, \Delta$}
\RightLabel{$\cut^*$}
\BinaryInfC{$D_1^3 \ldots D_1^{m_{1}}, \Gamma\Ra  D_2 \ldots D_n, \Delta$}
\noLine
\UnaryInfC{$\vdots$}
\noLine
\UnaryInfC{$\text{\fbox{$D_1^{m_{1}}$}},  \Gamma\Ra   D_2 \ldots D_n, \Delta$}
\RightLabel{$\cut$}
\BinaryInfC{$\Gamma\Ra D_2 \ldots D_n, \Delta$}
\DisplayProof
\end{footnotesize}
\end{center}
Next, for every $(D_3,\ldots,D_n)\in \Sigma_3\times\ldots\times \Sigma_n$, proceed as follows by cuts on the elements of $\Sigma_2$:
\begin{center}
\begin{footnotesize}
\AxiomC{\hspace{2cm}(\ref{next})\hspace{-2cm}}
\noLine
\UnaryInfC{$\Gamma\Ra \text{\fbox{$D_2^{m_{2}}$}}, D_3 \ldots D_n, \Delta$ \hspace{-4cm}}
\hspace{-0.1cm}
\AxiomC{(\ref{next})}
\noLine
\UnaryInfC{$\Gamma\Ra \text{\fbox{$D_2^2$}} D_3 \ldots D_n, \Delta$}
%
\AxiomC{(\ref{next})}
\noLine
\UnaryInfC{$\Gamma\Ra \text{\fbox{$D_2^1$}} D_3 .. D_n, \Delta$}
\AxiomC{(\ref{from-epsilon})}
\noLine
\UnaryInfC{$\text{\fbox{$D_2^1$}} D_2^2 .. D_2^{m_{2}}, \Gamma\Ra \Delta$}
\RightLabel{$\cut^*$}
\BinaryInfC{$\text{\fbox{$D_2^2$}}, D_2^3 \ldots D_2^{m_{2}}, \Gamma\Ra D_3, \ldots, D_n, \Delta$}
\RightLabel{$\cut^*$}
\BinaryInfC{$D_2^3 \ldots D_2^{m_{2}}, \Gamma\Ra   D_3 \ldots D_n, \Delta$}
\noLine
\UnaryInfC{$\vdots$}
\noLine
\UnaryInfC{$\text{\fbox{$D_2^{m_{2}}$}}, \Gamma \Ra  D_3 \ldots D_n, \Delta$}
\RightLabel{$\cut$}
\BinaryInfC{$\Gamma\Ra D_3, \ldots, D_n, \Delta$}
\DisplayProof
\end{footnotesize}
\end{center}
In this way we ultimately get $\Gamma\Sa\Delta$. Every introduced cut is on a formula from a critical antecedent $\Sigma_r$ in $\delta_1$ so by $\critical$ the cut is analytic or eliminable by induction hypothesis. Hence there is a locally analytic proof of $\Gamma\Sa\Delta$.

$\blacktriangleright$ $C$ is $A\coimp B$.
Similar to the previous case (by symmetry). 
This concludes the case distinction and thus the proof of the theorem. 
\qed
\end{proof}

\begin{theorem}
$\BiInt$ has the analytic cut property.
\end{theorem}
\begin{proof}
Starting with the uppermost one, all non-analytic cuts in a $\BiInt$-proof can be eliminated by repeatedly applying Theorem~\ref{thm:main}.\qed
\end{proof}

After an uppermost non-analytic cut in a proof is replaced with analytic cuts by the method of Theorem~\ref{thm:main}, it is in general \emph{not} true that their cut-formulas are subformulas of the endsequent. After all, there might be non-analytic cuts further down in the proof. In the course of restricting those lower non-analytic cuts, higher cuts will be revisited and modified again. It is only once the last non-analytic cut has been restricted that all cuts become ``globally analytic'' i.e.\ cut-formulas that are subformulas of the endsequent.

\section{The Core Ideas}
\label{attempt}
The above proof of the analytic cut property for $\BiInt$ is not tailored to this specific calculus. Here we identify the key properties that allow the proof to go through, and apply them to the case of the modal logic $\SF$. The first three properties are well-known to any proof theorist.

\begin{enumerate}
    \item \textbf{principal reductions:} replace a cut that is principal in left and right premise by cuts on smaller formulas
    (see Sect.~\ref{sec-idea}).
    
    \item \textbf{invertibility} of certain left and right rules: given a proof of the conclusion of a logical rule there are proofs of its premises with no new cuts.
  \item \textbf{axiom expansion:} completeness of
atomic axioms.  
\end{enumerate}
If the above properties hold for all connectives then cut-elimination is close: given a cut on a formula $C$ whose left and right rules are both invertible, replace the cut-formula in each premise with its immediate subformulas using invertibility and then apply the principal cut reduction. 

 All rules for connectives in $\BiInt$ admit principal reductions, \emph{including} $\imp$ and $\coimp$; e.g.,
    \[
    \infer[\cut]{\Gamma\Sa\Delta}
        {
        \infer[\coimp_R]{\Gamma\Sa A\coimp B,\Delta}
            {
            \Gamma,B\Sa\Delta
            &
            \Gamma\Sa A,\Delta
            }
        &
        \infer[\coimp_L]{A\coimp B\Sa\Delta}
            {
            A\Sa B,\Delta
            }
        }
    \]
    can be replaced by
    \[
    \infer[\cut]{\Gamma\Sa\Delta}
        {
        \Gamma\Sa A,\Delta
        &
        \infer[\cut^*]{\Gamma,A\Sa\Delta}
            {
            A\Sa B,\Delta
            &
            \Gamma,B\Sa\Delta
            }
        }
    \]
Moreover in the proof for $\BiInt$ we (implicitly) used the invertibility of $\land_L$, $\land_R$ and $\lor_L$,  $\lor_R$. 
We also used that $\top$ ($\bot$) can be removed from an antecedent (resp. succedent). 
Finally, axiom expansion holds for $\BiInt$ (Lemma~\ref{lem:genaxioms}).
However, the right rule of $\imp$ and the left rule of $\coimp$ are not invertible. 
What comes to the rescue are the following properties, here formulated for an arbitrary connective $\circ$:
\begin{enumerate}
  \setcounter{enumi}{3}
    \item \textbf{right-compatible:} Let $X$ and $Y$ be multisets of formulas that are permitted in the antecedent context  and succedent context  of the \textit{right rule} for $\circ$ respectively. For every rule instance $r$ of the calculus (excluding initial sequents) with non-empty antecedent resp.\ succedent context: appending $(X,Y)$ to the (antecedent,succedent) resp. (succedent,antecedent) everywhere yields a valid rule instance.
    
    \item \textbf{left-compatible}: Replace ``right rule'' in the above with ``left rule''.
\end{enumerate}

These properties relate the context restrictions of the introduction rule of a connective with the context restrictions of every other rule.
In particular, the motivation for `switching the contexts' from (antecedent,succedent) to (succedent,antecedent) is that the antecedent $\Sigma$ of the conclusion $\Sigma\Ra A\imp B$ of a critical inference must play a dual role, appending antecedent \textit{and} appending succedent of other rule instances. It is this that enables a picture as (\ref{cut-reduction-simple}).

\begin{example}
$\imp$ is right-compatible. 
The antecedent context of the rule $\imp_R$ can be instantiated by any multiset~$X$. The succedent context can only be instantiated by the empty multiset. We must ensure that appending $(X,\emptyset)$ to every rule instance with non-empty antecedent context  yields a rule instance (no need to check any $\coimp_L$ rule instance as it never has a non-empty antecedent context); also appending $(\emptyset,X)$ to every rule instance that has a non-empty succedent context (ruling out a rule instance of $\imp_R$) yields a rule instance.

Similarly $\coimp$ is left-compatible.

Another example ($\Box$ in $\SF$ is right-compatible) appears in Section~\ref{sec-SF}.
\end{example}

If a connective $\circ$ has a principal cut reduction, axiom expansion and a compatibility property, then a non-analytic cut on $C$ with main connective $\circ$ might be simplified along the following lines: trace $C$ upwards on each premise until principal, perform a principal reduction \emph{across different branches} (cf.\ (\ref{principal}) in main proof), then propagate the appended context downwards to the endsequent (relying on compatibility). Cut  (cf.\ Step 3 in main proof) each appended formula using a proof obtained (relying on compatibility) by the context appended to the `other side' i.e. antecedent appended to succedent.

To summarise, the argument works for $\BiInt$ because (a)~all connectives of $\BiInt$ admits principal reductions,  (b)~every connective of $\BiInt$ is either invertible $\{\top,\bot,\land,\lor\}$ or has some compatibility property $\{\imp,\coimp\}$, and (c) axiom expansion holds.

\subsection{An application: the case of $\SF$}\label{sec-SF}

A sequent calculus $\SF$~\cite{OhnMat57} for the modal logic $S5$ extends Gentzen's calculus $\LK$ for classical propositional logic with the following two rules for the modality~$\Box$:
\begin{equation}\label{SFrules}
\infer[T]{\Gamma,\Box A\Sa\Delta}
    {
    \Gamma,A,\Sa\Delta
    }
\qquad\qquad\qquad
\infer[5]{\Box\Gamma\Sa\Box A,\Box\Theta}     {
    \Box\Gamma\Sa A,\Box\Theta
    }
\end{equation}
As is well know, cut elimination fails but as shown by Takano~\cite{Tak92} the calculus
has the analytic cut property.
%
Cut-restriction provides a proof too. 

First notice that  $\SF$ admit principal cut reductions for all connectives, and axiom expansion. Also, the left and right rules of $\land$ and $\lor$ are invertible.
As the implication $A\imp B$ is classical, it can be considered as a derived connective and therefore removed from the signature. The role of the problematic connectives for cut-elimination ($\imp$
and $\coimp$ in $\BiInt$) is taken here by the modality $\Box$.

A boxed formula is principal in the antecedent by~$T$ (``left rule'') and in the succedent by~$5$ (``right rule'').
In the latter, the context consists of boxed formulas only, namely $\Box\Gamma$ in the antecedent and $\Box\Theta$ in the succedent. For every non-initial rule instance in $\SF$, appending arbitrary multisets of boxed formulas to the antecedent and succedent yields a new rule instance. 
It follows that $\Box$ is \textbf{right-compatible}. Once this is noted, essentially the same proof as for Theorem~\ref{thm:main} also works for $\SF$.
For illustration, we picture a simple case below.



\begin{center}
\begin{tabular}{c@{\hspace{1em}to\hspace{1em}}c}
\begin{footnotesize}
$
\infer[\cut]{\Gamma\Sa\Delta}
    {
    \infer*{\Gamma\Sa\Box A,\Delta}
        {
        \infer[5]{\Box D\Sa\Box A}
            {
            \Box D\Sa A
            }
        }
    &
    \infer*{\Gamma,\Box A\Sa \Delta}
        {
        \infer[T]{ \Sigma,\Box A\Sa \Pi}
            {
            \Sigma,A\Sa \Pi
            }
        }
    }
$
\end{footnotesize}
&
\begin{footnotesize}
$
\infer[cut]{\Gamma\Sa\Delta}
    {
    \infer*{\Gamma\Sa\Box D,\Delta}
        {
        \Box D\Sa\Box D
        }
    &
    \infer*{\Gamma,\Box D\Sa\Delta}
        {
        \infer[cut*]{ \Sigma,\Box D\Sa \Pi}
            {
            \Box D\Sa A
            &
            \Sigma,A\Sa \Pi
            }
        }
    }
$
\end{footnotesize}
\end{tabular}
\end{center}


\medskip
\noindent
{\bf Concluding Remark:}
We have investigated the first level of cut-restriction, adapting cut-elimination to obtain analytic cuts (when elimination is not possible). 
Our proof makes use of the structural rules of weakening and contraction. This leaves open the question of adapting it to substructural logics. 

%
\textit{What is the next level of cut-restriction to investigate?} Sequent calculi requiring mild violations of analytic cut, as in the case, e.g., of the calculi for the modal logics $K5$, $K5D$, and $S4.2$~\cite{Tak01,Tak19}. Here are some intuitions for $K5$. Takano~\cite{Tak01} observed that analytic cuts do not suffice but cuts on $K5$-subformulas do. These are subformulas of the endsequent, or $\Box\lnot\Box B$ or $\lnot\Box B$ where $\Box B$ is a subformula of some boxed formula that is a subformula of the endsequent. 

The sequent calculus for $K5$ 
extends $\LK$ with the following variant of $5$:
\begin{center}
\AxiomC{$\Gamma\Ra \Box\Theta, A$}
\RightLabel{$5^*$}
\UnaryInfC{$\Box\Gamma\Ra \Box\Theta, \Box A$}
\DisplayProof
\end{center}
Here is the main case (to simplify, assume one critical inference in each premise).
\begin{center}
\begin{footnotesize}
\AxiomC{$\Gamma'\Ra \Box\Theta',A$}
\RightLabel{$5^*$}
\UnaryInfC{$\Box\Gamma'\Ra \Box\Theta',\Box A$}
\noLine
\UnaryInfC{$\leftpremise$}
\noLine
\UnaryInfC{$\Gamma\Ra \Theta, \Box A$}
%
\AxiomC{$A \Ra \Box Y', B$}
\RightLabel{$5^*$}
\UnaryInfC{$\Box A \Ra \Box Y', \Box B$}
\noLine
\UnaryInfC{$\rightpremise$}
\noLine
\UnaryInfC{no $5^*$ rule on featured branch}
\noLine
\UnaryInfC{$\Box A, X\Ra Y$}
\RightLabel{$\cut^*$}
\BinaryInfC{$\Gamma, X\Ra \Theta, Y$}
\DisplayProof
\end{footnotesize}
\end{center}
The rule $5^*$ is not invertible, nor is it left/right-compatible: adding the same formula to the premise and conclusion antecedent breaks the rule instance (as the rule appends a box when passing from premise to conclusion antecedent). This suggests applying cuts on formulas that, in $\SF$, we propagated downwards. For simplicity, in the following, read each multiset as a single formula:
\begin{center}
\begin{footnotesize}
\AxiomC{$\Box\Theta'\Ra\Box\Theta'$}
\RightLabel{$\lnot_R$}
\UnaryInfC{$\Ra\Box\Theta', \lnot \Box\Theta'$}
\RightLabel{$5^*$}
\UnaryInfC{$\Ra\Box\Theta', \Box\lnot \Box\Theta'$}
  \RightLabel{$\weak$}
\UnaryInfC{$\Box\Gamma'\Ra\Box\Theta', \Box\lnot \Box\Theta'$}
\noLine
\UnaryInfC{$\leftpremise+$}
\noLine
\UnaryInfC{$\Gamma\Ra\Theta,\text{\fbox{$\Box\lnot\Box\Theta'$}}$}
%
%
\AxiomC{$\Box\Gamma'\Ra\Box\Gamma'$}
  \RightLabel{$\weak$}
\UnaryInfC{$\Box\Gamma'\Ra\Box\Theta',\Box\Gamma'$}
\noLine
\UnaryInfC{$\leftpremise+$}
\noLine
\UnaryInfC{$\Gamma\Ra\Theta,\text{\fbox{$\Box\Gamma'$}}$}
\hspace{-0.5cm}
\AxiomC{$\Gamma'\Ra \Box\Theta',A$}
\hspace{-0.5cm}
\AxiomC{$A \Ra \Box Y', B$}
\RightLabel{$\cut^*$}
\BinaryInfC{$\Gamma' \Ra \Box\Theta', \Box Y', B$}
\RightLabel{$\lnot_L$}
\UnaryInfC{$\Gamma' \lnot\Box\Theta'\Ra  \Box Y', B$}
\RightLabel{$5^*$}
\UnaryInfC{$\text{\fbox{$\Box\Gamma'$}}, \Box\lnot\Box\Theta'\Ra  \Box Y', \Box B$}
\RightLabel{$\cut^*$}
\BinaryInfC{$\Gamma,  \text{\fbox{$\Box\lnot\Box\Theta'$}}\Ra \Theta,  \Box Y',\Box B$}
%
\RightLabel{$\cut^*$}
\BinaryInfC{$\Gamma \Ra \Theta,\Box Y',\Box B$}
\noLine
\UnaryInfC{$+\rightpremise+$}
\noLine
\UnaryInfC{$\Gamma, X\Ra \Theta, Y$}
\DisplayProof
\end{footnotesize}
\end{center}
The cut on $\Box\Gamma'$ is like those we encountered before i.e. on a formula occurring in the original proof. There is also a cut on $\Box\lnot\Box\Theta'$ while it was $\Box\Theta'$ that occurred in the original proof but this is not unexpected (given Takano's result).

\newcommand{\critical}{\textbf{tameness}\xspace}

A significant issue remains: showing
that $\Box\Theta'$ 
is a subformula of the endsequent
(how to rule out a $K5$-subformula that is not a subformula?). A suitable irredundance property (and generalised \critical property) seems required.



%
\bibliographystyle{abbrv}
\bibliography{mybib.bib}
\end{document}